\documentclass[conference]{IEEEtran}
\ifCLASSINFOpdf
\else
\fi

\usepackage{graphicx}
\usepackage{subfigure}
\usepackage{sansmath}
\usepackage{amsthm}
\usepackage{float}
\usepackage{mathrsfs} 
\usepackage{algorithm}
\usepackage{algorithmicx}
\usepackage[noend]{algpseudocode}
\usepackage{url}

\theoremstyle{plain}
\newtheorem{theorem}{Theorem}
\newtheorem{lemma}{Lemma}

\theoremstyle{remark}
\newtheorem{remark}{Remark}

\theoremstyle{definition}
\newtheorem{definition}{Definition}

\newtheorem{assumption}{Assumption}
\newtheorem{example}{Example}

\usepackage{etoolbox}
\newtoggle{arxiv}

\toggletrue{arxiv}

\newcommand{\fullv}[1]{%
\iftoggle{arxiv}{%
#1
}
{}
}

\newcommand{\confv}[1]{%
\iftoggle{arxiv}{}
{#1}
}

\begin{document}
%
\title{A Scale-out Blockchain for Value Transfer with Spontaneous Sharding}



%
\author{\IEEEauthorblockN{Zhijie Ren\IEEEauthorrefmark{1},
Kelong Cong\IEEEauthorrefmark{2},
Taico V. Aerts\IEEEauthorrefmark{1}, 
Bart. A. P. de Jonge\IEEEauthorrefmark{1},
Alejandro F. Morais\IEEEauthorrefmark{1} and
Zekeriya Erkin\IEEEauthorrefmark{1}}
\IEEEauthorblockA{\IEEEauthorrefmark{1}Faculty of Electrical Engineering, Mathematics, and Computer Science\\
Delft University of Technology,
Van Mourik Broekmanweg 6, Delft, the Netherlands, 2628XE\\ Email: z.ren@tudelft.nl}
\IEEEauthorblockA{\IEEEauthorrefmark{2}Ecole Polytechnique F\'ed\'erale de Lausanne (EPFL), CH-1015 Lausanne, Switzerland\\
Email: kelong.cong@epfl.ch}
}


\maketitle

\begin{abstract}

Bitcoin, as well as many of its successors, require
the whole transaction record to be reliably acquired by all
nodes to prevent double-spending.
Recently, many blockchains have been proposed
to achieve scale-out throughput by letting nodes only acquire
a fraction of the whole transaction set. However, these schemes, e.g.,
sharding and off-chain techniques, suffer from a degradation in
decentralization or the capacity of fault tolerance.

In this paper, we show that the complete set of transactions is not a
necessity for the prevention of double-spending if the properties of value transfers is fully explored.
In other words, we show that
a value-transfer ledger like Bitcoin has the potential to scale-out by its nature without sacrificing security or decentralization.
Firstly, we give a formal definition for the value-transfer
ledger and its distinct features from a generic database. Then,
we introduce an off-chain based scheme with a shared main
chain for consensus and an individual chain for each node for recording transactions.
A locally executable validation scheme is proposed
with uncompromising validity and consistency.
A beneficial consequence of our design is that nodes will spontaneously try to reduce
their transmission cost by only providing the transactions needed to show that their transactions are double-spending-proof.
As a result, the network is sharded as each node only acquires part of the transaction record and a
scale-out throughput could be achieved, which we call ``spontaneous
sharding''.

\end{abstract}


%
\IEEEpeerreviewmaketitle

\section{Introduction}

Blockchain technology, made popular by Bitcoin \cite{nakamoto}, can be described as an append-only database maintained by distributed nodes instead of central authorities. One of the most well-known applications of blockchain technology is cryptocurrency, in which the blockchain is in the form of a distributed ledger, i.e., the data is transactions which are records of value transfers, called transactions, between nodes. The most crucial part of a distributed ledger for value transfer is the prevention of double-spending, which is achieved by consensus algorithms that guarantee all honest nodes in the network keep a consistent ledger of all valid transactions. The consensus algorithm can be divided into two categories, the Nakamoto-like consensus algorithms such as Proof-of-Work (POW) \cite{nakamoto} or Proof-of-Stake (POS) \cite{peercoin,algorand} and Byzantine fault tolerance (BFT) consensus algorithms such as PBFT \cite{castro}. For distributed ledger type of blockchain, most of the consensus algorithms effectively achieve the following conditions.
\begin{itemize}
\item {\bf Agreement (Consistency):} Two honest nodes should not have disagreement on the validity of a transaction.
\item {\bf Validity (Correctness):} Invalid transactions cannot be validated by honest nodes.
\item {\bf Termination (Liveness):} All transactions will be eventually known by all honest nodes.
\end{itemize}

Strictly speaking, the above conditions are not achievable in asynchronous networks \cite{flp,cap}. However, by slightly compromising either asynchronous \cite{castro} or deterministic conditions for termination \cite{benor,bracha}, the above-mentioned conditions can be achieved in practical asynchronous network.
Blockchains with both Nakamoto-like consensus \cite{garay} and BFT consensus can have scalable throughput, i.e., the communication cost per transaction (CCPT) is restricted to $O(N)$, where $N$ is the number of nodes in the network. Various consensus algorithms could achieve consensus with $O(N)$ complexity.
\begin{itemize}
\item {\bf Improved BFT algorithms:} Traditional BFT algorithms like \cite{castro,benor,bracha} have $O(N^2)$ CCPT. However, many recent BFT algorithms like \cite{cachin,miller,700bft,zyzzyva} achieve $O(N)$ CCPT by either packing up transactions or opportunistically running a much simpler scheme with traditional schemes as the back-up for the worst scenario.
\item {\bf Nakamoto-like consensus:} The POW scheme in Bitcoin introduced a game theoretical aspect to this problem. Then, instead of restricting the number of faulty nodes, an assumption is put on the rationality of nodes in the network. However, some early POW or POS based schemes have limitation in the transaction rate to meet the synchronous requirements \cite{croman}. With this problem solved in novel algorithms like \cite{algorand,ng,hybrid,ouroboros}, $O(N)$ CCPT is feasible in Nakamoto-like consensus.
\end{itemize}




\subsection{Scalability of Blockchain}

The most crucial problem in a decentralized value-transfer system is double-spending, which could be prevented when all nodes have a consistent record of all transactions. Then, $O(N)$ CCPT is required for all transactions. Blockchains with $O(N)$ CCPT are commonly referred as ``scalable'' blockchains since their throughput will not decrease (or increase) with the number of nodes and the computation and communication capacities in the network.

\subsubsection{Scale-out Blockchain Solutions}

Recently, several solutions have been proposed to achieve $o(N)$ CCPT, sometimes referred as ``scale-out'' throughput as the throughput will increase as $N$ grows, by reducing the number of validators and recordkeepers for each transaction. 
In other words, the termination property is compromised, i.e., a transaction is not necessarily known to or validated by the whole network, but a part of it.
Here, we introduce three types of such schemes.
\begin{itemize}
\item {\bf Off-chain Solutions:}
This type of approach are mostly associated with some existing blockchain systems as the main chain. Each node holds their transactions locally, sometimes referred as ``off-chain'', and only sends a description or the eventual outcome of these transactions to the ``main chain'', referred as ``on-chain''. Since there is no guarantee on the validity of the ``off-chain'' transactions, either validation nodes are introduced to validate and endorse these transactions \cite{rsk,polkadots} or economical deposit should be provided for the transactions \cite{lightning,plasma}. In both cases, the validity condition is compromised due to centralization or the economical constraint. 

\item {\bf Directed Acyclic Graph (DAG) Solutions:}
In another type of approach, we call DAG solutions, the transactions are not structured in a chain, but in a DAG \cite{tangle,byteball,swirld}. The validity is dependent on the (directly or indirectly) outgoing edges of the transaction, which represents the nodes that have validated it. A scale-out throughput can be achieved if the acquirement of the complete graph is not obligated for all nodes \footnote{In fact, in \cite{tangle,byteball,swirld}, the complete graph is required to prevent double-spending for their applications. Hence, they are not scale-out schemes, although a DAG scheme designed similarly could scale-out for some other applications.}. Then, the validity of the transactions is compromised due to its dependency on the validators.

\item {\bf Sharding Solutions:}
Recently, sharding solutions, which artificially divide the network, have been widely studied and discussed \cite{omniledger,shard,ethshard,chainspace}. They include schemes that fairly and randomly divide the network into small shards with vanishing probability of any shard having an overwhelming number of adversaries. Hence, the BFT consensus algorithm is run only within the shards and the CCPT is then $O(g^2)$ ($O(g)$ if scalable BFT algorithms are used) where $g$ is the size of the shard. However, the validity condition is also compromised in the sense that the sharding is only feasible when the ratio of adversaries in the network is small. Moreover, according to our knowledge none of the existing sharding schemes proved $g=o(N)$, which is the condition for scale-out throughput.
\end{itemize}

\subsubsection{Problem Statement}
It seems to be infeasible to achieve scale-out performance with the same level of security or decentralization as Bitcoin or blockchains using classical BFT algorithms. This does not come as a surprise since intuitively, double-spending can only be prevented with global consensus. This problem severely hampers the mainstream adoption of blockchain system since the security and trustworthiness of the blockchain system grows with the size of the network and the decentralization level. As a result, a trilemma is formed among throughput, security, and decentralization as also stated in \cite{omniledger,ethshard}.
However, at the meantime, traditional BFT algorithms could reach consensus on any type of message, which is redundant for many blockchain systems since messages in Bitcoin and cryptocurrencies are ``transactions'' which represent value transfers.

This leads to the research questions considered by this paper:
\begin{itemize}
\item What is the key functionality/features of the value-transfer blockchains?
\item Can we use these features to design a scale-out blockchain system to achieve the functionalities of value transfer without sacrificing reliability or decentralization?
\end{itemize}

\subsection{Overview of Our Solution}

Our solution gives an answer to the above questions. By exploring the features of value transfers which have not yet been used by other blockchain systems, we propose a blockchain system with a very simple structure to achieve scale-out throughput.

\subsubsection{Value-Transfer Ledgers}
Most of the aforementioned blockchain systems are decentralized solutions for value transfers and focus on reaching BFT consensus on transactions to prevent double-spending. However, traditional BFT consensus algorithms are generic and achieve BFT consensus regardless of the message type. In Bitcoin and other blockchain systems using Nakamoto-like consensus \cite{garay}, some realistic interpretation of transactions is used and the notion of rational behavior is introduced: rational issuers of transactions are interested in proving the validity of their transactions and keeping synchronized with other nodes. As a result, they either take effort themselves by mining or hire other nodes by paying transaction fee to submit their transactions to a blockchain which reaches BFT consensus.
In this paper, we take one step further to formally define the features for value transfers in the Value-Transfer Ledgers (VTL) model, i.e., 
\begin{itemize}
\item rational senders of the transactions should take effort to prove the authenticity of the transaction to the receiver;
\item rational receivers should check the authenticity of a transaction while receiving it;
\item a rational receiver will not care about the authenticity of other transactions unless they have an impact on their received transactions.
\end{itemize}
With these features, we propose a system that minimizes the redundancy of reaching BFT consensus on the transactions as if they are generic data and allows secure and reliable value transfers in a full decentralized fashion.
\subsubsection{Our System}
Our system has an off-chain structure, which contains individual chains for nodes to record their own transactions and a main chain for the consensus of the abstracts of their chains, i.e., provides a shared global state. Further, a locally executable validation function is proposed to have correct and consistent validation results upon all transactions. Besides a validation function, the crucial part of a valid validation scheme is that all honest nodes should also have a consistent observation of the transactions. In our system, we employ the aforementioned features of value transfers to achieved an alternation: instead of letting all nodes have consistent observation on all transactions, we guarantee that all nodes {\em that want to know the validity of a transactions} will have consistent observation on all transactions that {\em have impact on the validity of that transaction}. We also prove that this alternation is enough to have a valid system for value transfers.

\subsubsection{Spontaneous Sharding}

Moreover, the most innovative result in this paper is {\em spontaneous sharding}, which is a natural and direct consequence of using our system for value transfers. Generally speaking, in value-transfer systems, the values are passed from one node to another. In our system, for each piece of value, a proof is associated with it and the size of the proof grows with the number of nodes that it has been passed to. Then, since the sender could choose the source of his fund for the transactions, e.g., in Bitcoin, a node could choose from several of his unspent transaction outputs, rational nodes will choose the pieces of value with the minimum size of proof for the sake of the transmission cost. As a result, nodes will tend to cycle the value in small shards rather than the whole network. In other words, the network is sharded by the nature of the system without sacrificing either security or decentralization.

\subsection{Main Contributions}

The main contributions of this paper are the following.
\begin{itemize}
\item We formally define the VTL and distinguish it from other types of ledgers and databases.\footnote{A slightly similar idea have been raised in \cite{treechains} without a formal defined model or details for feasible schemes.}
\item We propose an off-chain based blockchain system that prevents double-spending without sacrificing either security or decentralization. In particular, our consensus algorithm achieves uncompromised agreement and validity conditions of the BFT consensus in VTL model.
\item We prove that our system is a valid VTL system. In other words, although our system do not guarantee BFT for generic types of data, we guarantee that if all nodes have interest in their values in the system and behave rationally, the valid transactions in our system are double-spending-proof.
\item It is shown that the CCPT of our system is upper bounded by $O(N)$, which suggests scalable throughput. Moreover, we show that our system could achieve scale-out throughput via spontaneous sharding in several scenarios.
\end{itemize}


\subsection{Content of This Paper}

This paper is organized as follows. In Section~\ref{s:pre}, we formally introduce the VTL model and assumptions. In Section~\ref{s:sys}, we introduce our system and prove the correctness of this system in VTL model. We analyze the performance of our scheme and introduce the concept of spontaneous sharding which results in scale-out throughput in Section~\ref{s:ana}. In Section~\ref{s:con}, we conclude our paper with possible topics for further exploration.


\section{Model}\label{s:pre}


In this paper, we emphasize on our novelties and contributions by showing that our system provides the minimum functionalities for value transfers. These functionalities can be used as building blocks for more generic VTL systems. Hence, some other elements are simplified to the most comprehensive level, e.g.,
\begin{itemize}
\item We consider every node holding some initial value. The mining of new coins is not considered.
\item Transactions are defined similarly to Bitcoin, namely the Unspent Transaction Output as input (UTXO) structure. We assume a transaction has only one sender and one receiver.
\item We consider a weak asynchronous network with $f \leq \lfloor \frac{N-1}{3} \rfloor$ Byzantine adversaries, just as the one used in \cite{castro}, so that PBFT can be straightforwardly applied. Note that this assumption is solely made for easy comprehension of our system. The same framework proposed in our system can be plugged into any permissioned or permissionless blockchain or consensus algorithm that achieves global BFT consensus on all transactions.
\item We assume that there exists an unbreakable hash function $Y=H(X)$ and a digital signature scheme $Y=Sig_{i}(X)$ based on the public-private key pairs where node $i$ is the signer.
\end{itemize}





\subsection{Network Model}

We consider a weak asynchronous network of $N$ nodes in which the message delay does not increase indefinitely as described in \cite{castro}. Each node holds some initial value that could be transacted with others. We assume that there are $f \leq \lfloor \frac{N-1}{3} \rfloor$ Byzantine adversaries and we have the following definitions for honest nodes and adversaries.
\begin{definition}[Honest nodes and Adversaries]
{\bf Honest nodes} will follow the schemes of the system.
{\bf Adversaries} can behave arbitrarily. 
\end{definition}

The network is assumed to be permissioned, i.e., the nodes are known to each other by their identities $n \in \{1, 2 ,\ldots,N\}$. We also assume that there exists a public key infrastructure (PKI) and nodes can link between the identity and the public key of each node.
Moreover, we introduce the ``chain'' as a data structure that consists of an ordered sequence of blocks. Each block consists of multiple transactions and a hash digest of the previous block, except for the first block, namely the genesis block.

\subsection{VTL Model}

Inspired by Bitcoin, most of the blockchain systems mimic value transfer systems, e.g., currency, in a decentralized fashion. The problem could be described as the following: Each node holds some positive value that they could transfer to others via transactions. A transaction is only valid if it is authorized by the owner of the value and the value cannot be double-spent. In other words, for any ``value'' in the network, it has three properties:
\begin{enumerate}
\item {\bf Ownership:} Value has an owner. Only the owner of the value can authorize to transfer his value.
\item {\bf Fluidity:} Any transfer can be completed in finite time.
\item {\bf Validity:} The value cannot be created or duplicated.
\end{enumerate}
A value transfer system can be in many forms, e.g., the account-based ledgers, which is widely used in banking system and many other accounting systems. Bitcoin, as well as many other blockchain systems, use a ledger with UTXO structure, which is very suitable for decentralized systems. Here, we introduce the UTXO structure.

\subsubsection{UTXO}

Firstly, in UTXO a transaction is an authorized piece of information that transfers the value from one node to another. In this paper, we use the following definition for a transaction, which is a slight variation of the traditional UTXO structure used in Bitcoin.
\begin{definition}[Transaction]
\sloppy A transaction $tx_i$ is a five-tuple:
$tx_i= \langle {\rm Source}_i, s_i, d_i, a_i, r_i \rangle$
where ${\rm Source}_i$ is the set of transactions which are used as the source, $s_i$ is the sender, $d_i$ is the receiver, $a_i$ is the transacted value, and $r_i$ is the remaining value.
\end{definition}

In Bitcoin, transactions are usually referred to the ones on the longest chain, which also suggest that they are {\em valid} transactions. However, in some other systems like \cite{ethereum,spectre}, the concept of a valid transaction is ambiguous since invalid transactions can also exist on the chain. As a result, a deterministic and consistent rule should be applied for all nodes to determine the valid transactions. Here, we define the valid transaction in UTXO as the following.

\begin{definition}[Validity of a Transaction]\label{def:valtx}
\sloppy A transaction $tx_i=\langle {\rm Source}_i, s_i, d_i, a_i, r_i \rangle$ is valid if and only if the following conditions hold.
\begin{itemize}
\item {\bf Confirmed and authorized:} $tx_i$, as well as some witness indicating that $tx_i$ is authorized by $s_i$, e.g., a digital signature of $tx_i$ signed by $s_i$, are on a tamper-proof ledger.
\item {\bf Valid sources:} All transactions in $tx_j \in {\rm Source}_i$ are valid.
\item {\bf Value equality:} The original value equals to the sum of the transacted value and the remaining value, i.e., $\sum_{tx_j \in {\rm Source}_i} r_j = a_i+r_i$.
\item {\bf No double-spending:} Assuming $tx_i=t_{u,k,l}$, for any $tx_j \in {\rm Source}_i$, there does not exist a valid transaction $tx_{i'}= t_{u,k',l'}, tx_j \in {\rm Source}_{i'}$ such that $k' < k$ or $k\ = k, l' < l$.
\end{itemize}
\end{definition}

Then, we define the unspent transaction in UTXO.
\begin{definition}[Unspent Transaction]
A transaction $tx_i$ is an unspent transaction if there is no other transaction $tx_j$ in the ledger which is valid and $tx_i \in {\rm Source}_j$.
\end{definition}

Clearly, in UTXO structure, the value exists in the form of unspent transactions. The value is always transferred from one unspent transaction to another unspent transaction, instead of transferring from one account to another account as the account-based ledger structure.
\subsubsection{Properties of VTL Model}

Traditional blockchain systems prevent double-spending by reaching the classical BFT consensus on all transactions. The most straightforward approach is to treat transactions as bit strings and use classical BFT algorithms \cite{castro,benor,bracha} or improved BFT algorithms \cite{cachin,miller,zyzzyva} to reach consensus. However, this approach misses the notion of ``value'' behind the transactions and discards the differences between a transaction and general data. These differences will be uncovered if the original notion of value is focused. Here, we pick up the idea behind the ``rational nodes'' and ``transaction'' notions in Nakamoto-like consensus and add more real-world interpretations to these two notions in value transfers.

Firstly, value has an owner and the receiver of an unspent transaction is the owner of that value until it is spent again. The owner would take full initiative and responsibility of proving the existence and the authenticity of the value to any node upon request. If he fails to do so, it will be considered as against his own interest. For instance, if the value is considered as money, the holder of the money is motivated to prove the money is real when he uses it for purchase. A failure in proving will cause the purchase to fail, which is against his own interest.

Secondly, the concern of the nodes is the authenticity of the value they owned instead of the transaction records. Hence, nodes will check the past transaction records only if the records have impact on the authenticity of the value that they concern. Otherwise, nodes have no interest and will not care about the validity of a past transaction.

As a result, we make three assumptions in VTL model. Throughout this paper, we use the term ``node $u$ is curious about transaction $tx_i$'' to represent that node would like to check the validity of transaction $tx_i$.



\begin{assumption}[History Disinterest]\label{pro:hd}
A node $u$ is curious about a spent transaction $tx_i$ only when it is curious about an unspent transaction $tx_j$ and the validity of $tx_i$ is required to check the validity of $tx_j$.
\end{assumption}
\begin{assumption}[Rational Receiving]\label{pro:rr}
A node $u$ is curious about transaction $tx_i$ if it is the receiver of $tx_i$ and does not know the validity of it.
\end{assumption}
\begin{assumption}[Rational Owner]\label{pro:ro}
If node $u$ is the receiver of a valid and unspent transaction $tx_i$, it will provide the validity proof of $tx_i$ to any node once it is requested.
\end{assumption}

In practice, it is not rational for a node to validate an unspent transaction if it is not the receiver since validation is resource consuming. Hence, we have an alternative version for Assumption~\ref{pro:rr} to minimize the cost in a resource-limited network.

\begin{assumption}[Rational and Cost-saving Receiving]\label{pro:rcsr}
A node $u$ is curious about transaction $tx_i$ {\bf if and only if} it is the receiver of $tx_i$ and does not know the validity of it.
\end{assumption}

This assumption will not affect the correctness of our scheme. It will be applied in the performance analysis for simplicity.

\subsubsection{Valid VTL System}

Then, we define a valid VTL system with a structure of UTXO.
\begin{definition}[Valid VTL System]\label{def:vvtls}
A system with UTXO structure is called a valid VTL system if it satisfies the following conditions under Assumption~\ref{pro:hd}-\ref{pro:ro}:
\begin{enumerate}
\item {\bf Ownership:} If an honest node receives a valid transaction, then he can make one valid transaction using it as a source. Meanwhile, no other node can make a valid transaction using it as a source.
\item {\bf Fluidity:} A valid transaction will be considered as valid by the receiver in finite time if both the sender and the receiver are honest.
\item {\bf Validity:} Invalid transactions will not be considered as valid by honest nodes.
\end{enumerate}
\end{definition}

\begin{remark}[Relationship Between BFT Consensus and Valid VTL System]\label{rm:bftvtl}
Clearly, the Validity condition of valid VTL systems is exactly the Validity condition in the BFT consensus. Then, the Fluidity condition is guaranteed by all three BFT consensus conditions and the Ownership condition is guaranteed by the Agreement condition and the way that UTXO structure is designed. Hence, BFT consensus on the transactions is a sufficient condition for a valid VTL system. However, later we will show that it is not a necessary condition for VTL since the BFT consensus with a weakened Termination condition is also sufficient for a valid VTL system.
\end{remark}

\section{Our System}\label{s:sys}

Our system consists of three parts: individual chains for transactions, a main chain for a global shared state, and a validation scheme for validation of the transactions.
In this section, we first introduce these three parts of our system and give important theorems of the system. Then, we prove that our system is a valid VTL system as well as prove that our system actually only compromises the Termination condition of the BFT consensus.



\subsection{Individual Chains}\label{ss:ic}
\sloppy Each node generates an individual chain to record their own transactions in a first-in-first-out fashion. An individual chain of node $u$ is an ordered set of blocks $\{B_{u,1},B_{u,2}, \ldots, \}$ and a block is an ordered set $B_{u,k}=\{H(B_{u,k-1}), t_{u,k,1},t_{u,k,2}, \ldots\}$, where $t_{u,k,l}$ is a transaction sent by node $u$ with valid sources, value equality, and no double-spending as defined in Definition~\ref{def:valtx}. In our system, we assume that there is an initial value assigned to each node in the same fashion as a transaction with no source. The sender and receiver of this transaction are both the node itself.

The size of a block can be arbitrary. Periodically, nodes send \emph{Abstracts} to the \emph{main chain} (will be introduced in the next paragraph). The abstract is defined as the following.
\begin{definition}[Abstract]
\sloppy An abstract of block $B_{u,k}$, denoted by $A_{u,k}$, is a four-tuple: 
$A_{u,k} = \langle u, k, H(B_{u,k}), Sig_u(u||k||H(B_{u,k})) \rangle$.
\end{definition}

\subsection{Main Chain}
The main chain uses PBFT as its consensus algorithm and the blocks consist of \emph{Abstracts} signed by the corresponding nodes. We assume that the abstracts of all genesis blocks are on the main chain. Since it has been proved that the PBFT can reach BFT consensus on messages in our network model \cite{castro}, we simply see the main chain as a reliable and secure primitive and all abstracts included on the main chain reaching the BFT consensus. Honest nodes will send abstracts of their newest blocks to the main chain when they observe that their previous abstracts are on-chain.

\subsection{Confirmation}\label{ss:conf}

The transactions on individual chains are arbitrary in the sense that they are neither tamper-proof nor signed. The transactions will be tamper-proof and signed if an abstract of a block that comes after it is contained in the main chain, which we call confirmed transactions. Here, we give the formal definitions of a confirmed transaction and a confirmed block.
\begin{definition}[Confirmation]
\sloppy A block $B_{u,k}$ is confirmed if 
\begin{itemize}
\item an abstract of the block or a block after it, i.e., $A_{u,k'}, k'\geq k$, is on the main chain;
\item for all abstracts of node $u$, denoted by $A_{u,l}$, that are on the main chain and $l \leq k'$, $A_{u,l}$ is compliant to their corresponding blocks.
\end{itemize}
A transaction $tx_i=\langle {\rm Source}_i, s_i, d_i, a_i, r_i \rangle$ is a confirmed transaction if $tx_i \in B_{u,k}, s_i = u$, and $B_{u,k}$ is confirmed. We call $B_{u,k}$ and $tx_i$ are confirmed by $A_{u,k'}$.
\end{definition}

The confirmation of a transaction suggests that it is tamper-proof as if it is on-chain, which is shown in the following theorem.
\begin{theorem}[Confirmed Transactions]\label{pro:ct}
If $tx_i=t_{u,k,l}$ is a transaction confirmed by abstract $A_{u,k'}, k' \geq k$, then there does not exist a chain of confirmed blocks $\{B'_{u,1},B'_{u,2}, \ldots, B'_{u,k'}\}$ such that $t'_{u,k,l} \ne tx_i$ and all hashes are correct. 
\end{theorem}
\confv{The formal proof of this theorem can be found in the extended version in \cite{ren}.}\fullv{The formal proof of this theorem is given in Appendix~\ref{a:1}.} By Theorem~\ref{pro:ct}, when a transaction is confirmed, the position and content of it cannot be changed. Furthermore, it is also signed since the sender of the transaction is the same as the signer of the abstract and the abstract contains an unforgeable signature of the sender. Note that a confirmed transaction here is not the same as confirmed transaction in other blockchain systems like Bitcoin, as they are not yet validated.

\subsection{Validation Scheme}\label{s:val}

Our validation scheme consists of two parts: a proof collection process that allows any node that is curious about a transaction to reliably and efficiently collect the proof of it; a validation function that deterministically decide whether a transaction is valid or not depending on the collected proof.


\subsubsection{Proof Collection}

First we define the validity proof of a transaction $tx_i$.

\begin{definition}[Validity Proof]\label{def:vp}
Assuming that the sender $s_i=u$ for transaction $tx_i$, $tx_i \in B_{u,k}$, and there exists an abstract $A_{u,k'}, k' \geq k$ in the main chain, a validity proof ${\cal P}(tx_i)$ is the union of a set of all blocks before and including $B_{u,k'}$ and the proofs of all transactions in ${\rm Source}_i$, i.e., ${\cal P}(tx_i) = \{B_{u,k''}| k'' \leq k'\} \cup \{B_{v,l}| B_{v,l} \in {\cal P}(tx_j), tx_j \in {\rm Source}_i\}$.
\end{definition}

By Definition~\ref{def:vp}, a validity proof of $tx_i \in B_{u,k}$ includes the chain of $u$ from the genesis block to a block $B_{u,k'}, k' \geq k$ which has an abstract in the main chain. Moreover, it also includes the chains of the sources of this transaction, and recursively the sources of the sources until the genesis block.

In the following lemma, we show that the proofs of valid transactions can always be collected by nodes who are curious about them in the VTL model.
\begin{lemma}[Feasibility of the Proof Collection]\label{pro:fpc}
If a node $u$ is curious about a valid transaction $tx_i$, then it can always identify a node $v$ such that it would provide the proof of $tx_i$.
\end{lemma}
\begin{proof}
If $tx_i$ is an unspent transaction, this lemma directly follows from Assumption~\ref{pro:ro} since the receiver of $tx_i$ will provide it. 
If $tx_i$ is a spent transaction, then by Assumption~\ref{pro:hd}, $u$ will only be curious about $tx_i$ if $u$ is curious about an unspent transaction $tx_j$ and the validity of $tx_i$ is required for the validity of $tx_j$. By Definition~\ref{def:vp}, we have ${\cal P}(tx_i) \subset {\cal P}(tx_j)$. Hence, by Assumption~\ref{pro:ro}, $u$ can collect the proof of $tx_j$ from the receiver of $tx_j$.
\end{proof}
By Lemma~\ref{pro:fpc}, the proof of a transaction $tx_i$ can always be collected reliably and efficiently. By reliably, we mean that by Lemma~\ref{pro:fpc}, the proof can always be collected in the VTL model without any risk of disconnections. By efficiently, we mean that the collection is a simple point-to-point communication without the need of a reliable broadcast scheme like \cite{bracha} to tolerant malicious behaviors. 

However, although by our model the receiver of an unspent transaction is motivated to provide the correct proof, the requester of the proof will not accept it as a proof without his own verification. \confv{The algorithm is omitted here for simplicity and can be found in \cite{ren}.}\fullv{We give the Proof Verification Algorithm $\mathsf{Ver}({\cal P}(tx_i))$ in Algorithm~\ref{alg:pv} in Appendix~\ref{a:2}.}
If $\mathsf{Ver}({\cal P}(tx_i))=\mathsf{pass}$, it suggests that ${\cal P}(tx_i)$ is indeed a validity proof for transaction $tx_i$ since the algorithm is a direct translation from the definition of the validity proof.

\subsubsection{Validation Function}
The deterministic Validation Function is given in Algorithm~\ref{alg:vf}.

\begin{algorithm}
\caption{Validation Function $\mathsf{V}(tx_i, {\cal P}(tx_i)), tx_i= \langle {\rm Source}_i ,s_i, d_i, a_i, r_i \rangle \in B_{u,k}$}
\label{alg:vf}
\begin{algorithmic}
\State \#{Validity Proof Check}
\If{$\mathsf{Ver}({\cal P}(tx_i)) \ne \mathsf{pass}$} \Return $\mathsf{unknown}$ \EndIf
\State \#{Equality Check}
\If{$\sum($all remaining values from ${\rm Source}_i) \ne a_i+r_i$} \Return $\mathsf{unknown}$ \EndIf
\State \#{Double-Spending Check}
\For{$B_{u,m}, m =[1:k]$}
	\For{All transactions $tx_j$ in $B_{u,m}$}
		\If{${\rm Source}_j \cap {\rm Source}_i \ne \emptyset$ {\bf and} $tx_i \ne tx_j$} \Return $\mathsf{unknown}$ \EndIf
	\EndFor
\EndFor
\State \#{Source Check}
\For{all transactions $tx_j$ in ${\rm Source}_i$}
	\If{$\mathsf{V}(tx_j, {\cal P}(tx_j)) \ne \mathsf{valid}$} \Return $\mathsf{unknown}$ \EndIf
\EndFor
\Return $\mathsf{valid}$
\end{algorithmic}
\end{algorithm}



The correctness of the validation function is given in the following theorem.
\begin{theorem}\label{th:vs}
$\mathsf{V}(tx_i, {\cal P}(tx_i)) = \mathsf{valid}$ if and only if $tx_i$ is valid.
\end{theorem}

This theorem holds since the validation function is a straightforward translation of the definition of the validity. \confv{The detailed proof is given in \cite{ren}.}\fullv{The detailed proof is given in Appendix~\ref{a:4}.}

\subsection{BFT Satisfactory}

Here, we show that our system satisfies the agreement and validity condition of BFT  with a compromised termination condition for all valid transactions.

\begin{theorem}[BFT Satisfactory]\label{th:bft}
Our system satisfies the following conditions in VTL model. Here, we use the term ``node $u$ validates a transaction $tx_i$'' to represent that node $u$ runs a validation function on $tx_i$ with the result $\mathsf{valid}$.
\begin{itemize}
\sloppy \item {\bf Agreement (Consistency):} If an honest node validated a transaction, then, if another honest node is curious about this transaction, it will also validate it.

\item {\bf Validity (Correctness):} 
If a transaction can be validated by an honest node, then at least one honest node that is curious about it can validate it.

\item {\bf Termination (Liveness):} If a transaction is proposed by an honest node, then it can be validated in finite time.
\end{itemize}
\end{theorem}

\begin{proof}
$ $

\begin{itemize}
\sloppy \item {\bf Agreement:} If a transaction $tx_i$ is validated by an honest node, i.e., $\mathsf{V}(tx_i, {\cal P}(tx_i)) = \mathsf{valid}$, then, by Theorem~\ref{th:vs}, $tx_i$ is valid. By Lemma~\ref{pro:fpc}, if another node is curious about $tx_i$, the proof can be collected. Then, since the validation function is deterministic, another curious honest node will also run the validation function and the result will be $\mathsf{valid}$.

\item {\bf Validity:} 
Firstly, by Theorem~\ref{th:vs}, a validated transaction is equivalent to a valid transaction. Then, by Lemma~\ref{pro:fpc} its proof can be collected by an honest curious node and by Theorem~\ref{th:vs} it will be validated.



\item {\bf Termination:} If a transaction is proposed by an honest node, by the definition of the honest node, it should have valid sources, value equality, and no double-spending as suggested in Subsection~\ref{ss:ic}. Then, by the BFT satisfactory of the PBFT scheme we used for the main chain, this transaction will eventually be confirmed and meets all requirement of a valid transaction. Then, by Theorem~\ref{th:vs}, it can be validated. 
\end{itemize}
\end{proof}


\subsection{Validity of the System}
Now we show that our system is a valid VTL system by showing all three conditions in Definition~\ref{def:vvtls} are guaranteed in the VTL model.
\begin{theorem}[Validity of Our System]\label{th:vtl}
A system described in this section is a valid VTL system.
\end{theorem}
\begin{proof}
$ $
\begin{itemize}
\item {\bf Ownership:} By the Validity condition of Theorem~\ref{th:bft}, honest nodes can make valid transactions with valid received transactions as sources.
Meanwhile, if another node makes a transaction with sources that he is not the receiver, by Theorem~\ref{pro:ct}, the proof will be considered as incorrect and the validation will fail.
\item {\bf Fluidity:} This condition is guaranteed by combining all three conditions in Theorem~\ref{th:bft}.
\item {\bf Validity:} This condition directly follows from the Validity Condition in Theorem~\ref{th:bft}.
\end{itemize}
\end{proof}
The insight of Theorem~ref{th:vtl} showing our system as a valid VTL system is that our system does not guarantee BFT for generic data and cannot be used in applications where Assumption~\ref{pro:hd}-\ref{pro:ro} do not hold. For example, valid transactions cannot be validated by any other nodes if the sender refuses to offer the proof to any other nodes. However, this scenario is basically denying the value of the sender himself and thus should not happen in the VTL model with UTXO structure if the nodes are rational. On the other hand, if the transactions are in the form of debits instead credits and the receiver is the interested party, our system could guarantee neither the BFT condition nor the conditions for a valid VTL system.

\section{Performance Analysis and Spontaneous Sharding}\label{s:ana}

In this section, we will give explanations for the scale-out claim that we made for the throughput. First we show that the throughput of our system is scalable even in the worst case and will naturally scale out if the transaction pattern is already sharded. Then, we explain why and how our system could spontaneously shard. We give examples with theoretical and simulative analysis to show that the feasibility of spontaneous sharding as well as the scale-out throughput.


\subsection{Communication Cost Per Transaction}\label{ss:ccpt}

In our system, the main chain is using PBFT with $O(N^2)$ message complexity. However, the number of transactions associated with one abstract in the main chain are arbitrary and independent of the main chain. As a result, the communication cost of the main chain can be made into a negligible term in CCPT if we choose the number of transactions associated with one abstract to be $\omega(N^2)$. The duration of the consensus process still plays an important role in the latency of our system. However, note that the PBFT-based scheme is used only for easy comprehension and can be easily replaced by other scalable and low latency blockchain systems to improve the latency. 

Then, we make a few assumptions for easier analysis of the CCPT. We assume that rational nodes will not care about invalid transactions and thus will not try to re-collect the proof of a transaction if it is failed for a number of times. In other words, malicious nodes cannot spam invalid proofs to jam the network. For the simplicity in analysis, we only consider the resource limited network with Property~\ref{pro:rcsr}. Then, in our system, the CCPT can be represented by $p/T$, where $p$ is the total communication cost of all proofs and $T$ the total number of transactions made by the whole network. 

In general, the proof of a transaction $tx_i$ includes the chains of the sender, the senders of all sources of this transaction, and the senders of, recursively, the sources of the sources. In most blockchain systems, the storage is traded for validation efficiency, i.e., the validated transactions and their proofs are stored and the proofs of new transactions are collected incrementally. 
Then, in the worst case when the proof of any transaction includes the chains of all nodes, if the storage is not limited, all nodes simply need to be updated with all transactions in the network.  Each transaction will be communicated by $O(1)$ message per each node due to Property~\ref{pro:ro} and Lemma~\ref{pro:fpc} since a point-to-point based collection is sufficient to guarantee reliability and there is no need for BFT reliable broadcast schemes. Then, for each node, the communication cost for all its proofs is $O(1)T$. For the whole network, $p=O(1)NT$ and the CCPT is thus $O(N)$.

A better case would be that the transaction pattern is separated into shards and the nodes only make intra-shard transactions. In that case, the proof of any transaction contains the chains of only the nodes in their shards and the CCPT is $O(g)$, where $g$ is the size of the shard. As a result, our system achieves scale-out throughput.

\subsection{Spontaneous Sharding}\label{ss:ss}

Here, we consider a more interesting case that the transaction pattern is not separated into small shards and show that our system could still achieve scale-out throughput if all nodes behave rationally. We call this spontaneous sharding. The idea behind the spontaneous sharding is simple: rational nodes will try to minimize their transmission and storage costs by minimizing the proof size of each transaction as well as the number of recorded transactions. Then, the minimum cost of either transmitting or storing the transaction is actually the cost of proving the authenticity of the value in the transaction, which depends on the number of nodes that this value has been passed through.

More precisely, let us focus on a piece of value originated from a genesis block. For each time that it is transferred to a node, a confirmed chain of that node is included in the proof of the transaction of that value. Meanwhile, if the value is used together with other sources to make transaction, then these pieces of value are combined as well as their proofs. Hence, rational node will always avoid combining values for the sake of transmission and storage costs. Besides, the cost can be reduced by trying to make each piece of value only cycling in a small shard of the network, which makes the proof of the value only contain the chains of the nodes in that shard. Let us denote $P_i$ for the set of chains that are included in the proofs of all values owned by node $i$. By the analysis made in Subsection~\ref{ss:ccpt}, there is minimum overhead for proof collection in our system and the communication cost for each transaction per node is $O(1)$ messages. Hence, the CCPT in our system can be calculated as $O(g)$, where $g=E[|P_i|]$. This is called spontaneous sharding since the throughput increased as if the network has been naturally sharded.

Here, we give a more detailed analysis with graph theory. Let us consider the transaction pattern as a weighted directed graph $G(V,\vec{E})$, where the vertices $v \in V$ represent the nodes in the network. Then, instead of actually transactions, the edges $e=(u,v,w), e\in \vec{E}$ represent the transaction channels and their capacities, i.e., the existence of transactions between the sender $u$ and receiver $v$, and the rate for the transactions (amount per second) denoted by $w$. Then, we assume a stable and sustainable value-transfer network, where all nodes have equal amount of income and outcome in a long term. Let us denote the sets for inbound edges and outbound edges of node $i$ as $I_i$ and $O_i$, respectively, i.e., $I_i= \{e \in \vec{E} : e=(u,i,w)\}$, $O_i= \{e \in \vec{E} : e=(i,u,w)\}$.

When a piece of value sent by node $i$ via edge $e=(i,u,w) \in O_i$ to node $u$, this amount of value $w$ should return to node $i$ through a path. The nodes on the path form a node set, denoted by ${\cal N}_e \subseteq V$. In our system, the value will then return with all the chains of nodes in ${\cal N}_e$. The same holds for the values received from the edges in $I_i$ as each piece of these values will eventually return to the corresponding inbound neighbor of node $i$. As a result, we have $P_i=\{v \in {\cal N}_e : e \in I_i \cup O_i\}$.

There are two ways to optimize $P_i$ and achieve spontaneous sharding: local optimization and global optimization.

A local optimization can straightforwardly be done by the following:
according to the information about the chains that the receiver already has, the sender will choose from all its unspent transactions for the ones with the least amount of required proofs. For example, node 1 has transacted with a receiver node $2$ who has already acquired the chains of $\{3,4,5,6\}$ in this round. Then, if node 1 has this information, it will prefer to use the unspent transactions with proofs that consist of the chains from $\{3,4,5,6\}$ for transactions to node 2 so that it does not need to send proofs anymore. \confv{We give a naive smart transacting algorithm in \cite{ren}.}\fullv{In Appendix~\ref{a:3}, a naive smart transacting algorithm is given to achieve local optimization.}

A global effort could be made by letting all nodes broadcast their acquired chains each consensus round. This effort is spontaneous and beneficial to the nodes themselves, so a reliable broadcasting scheme is not necessary. An additional $O(N^2g)$ communication cost is required each round, which adds at most $O(g)$ to the CCPT. With the global information, some optimization schemes could be run locally as references for the source selection when nodes send transactions. The purpose of the global optimization is to route all values with the same sink through paths that include the minimum number of nodes. A feasible global optimization scheme is a non-trivial problem that we leave for future research.

However, the performance of the sharding and the eventually throughput depends heavily on the network model and the transaction pattern. Here, we give theoretical analysis showing the possibility of our system to scale-out in some large random networks. Then, we use simulation to show spontaneous sharding is feasible even in small networks.

\begin{remark}[Tragedy of the Commons]
It seems that the spontaneous sharding would only happen if all nodes perform rationally and cooperate, which will fall into the pitfall of tragedy of the commons \cite{hardin} if some nodes with high capacity do not optimize their proof sizes. However, this system is not identical to the tragedy of the common scenario since spontaneous sharding could also happen locally so that transmission cost is a private resource rather than public resource. In other words, a group of resource limited nodes can optimize their transactions locally and reduces their transmission cost without needing global cooperation.
\end{remark}

\subsection{Examples}

Firstly, we show that our system with a global optimization scheme will scale out in large random networks.

\begin{example}\label{ex:1}
We consider a random directed weighted graph constructed as the Erd{\H{o}}s-R{\'e}nyi model with $N$ nodes, $M$ edges, connectivity $p=\frac{M}{N(N-1)}$, and $p$ is larger than $\frac{\ln N}{N}$ so that the network is fully connected. We define $f$ as the ``weight factor'', which is the average number of transaction channels required for a piece of value, i.e., $f= E [ \lceil \frac{w}{E[w]} \rceil]_w$. This value is 1 in unweighted graph and is $O(1)$ if $w$ follows Poisson distribution.
\end{example}

For Example~\ref{ex:1}, let us consider the average size of $P_i$. Firstly, if our system is globally optimized, we will have
\begin{eqnarray}
E[|P_i|] &=& E[|\{v \in {\cal N}_e : e \in I_i \cup O_i\}| \\
&\leq & E[\sum_{e\in I_i \cup O_i} |{\cal N}_e|]. \label{eq:1}
\end{eqnarray}
Hence we focus on one path in which the value from an outbound edge $e$ from node $i$ flows back to $i$. The average path length is denoted by $l$. Then, due to the limited capacity of the edges, this value requires $\sim fl$ edges to be delivered. Hence, there are $\sim fl$ nodes in the set ${\cal N}_e$. Combining this with (\ref{eq:1}) we have
\begin{eqnarray}
\lefteqn{E[\sum_{e\in I_i \cup O_i} |{\cal N}_e|]} \\
&\sim & 2cfl \label{eq:12} \\ 
& \sim & 2f \cdot c \frac{\ln N}{\ln pN}  \label{eq:2} \\
& \sim & 2f \cdot c \frac{\ln N}{\ln c}, \label{eq:3}
\end{eqnarray}
where $c$ is the average inbound (outbound) connectivity, which equals to $M/N=p(N-1)$. Here, we have (\ref{eq:12}) since the average number of elements in $I_i \cup O_i$ is $2c$. Then, (\ref{eq:2}) follows from the average path length $l \sim \frac{\ln N}{\ln pN}$ from the random graph. Combining (\ref{eq:3}) with (\ref{eq:1}) and observing that $2f \cdot c \frac{\ln N}{\ln c}$ is dominated by $N$ when $c=o(\frac{N}{f\ln N})$, we have the following condition: if conditions
\begin{equation}
\left\{ \begin{array}{l} p > \frac{\ln N}{N} \\ p = o(\frac{1}{f\ln N}) \end{array} \right.
\end{equation}
hold, our system could scale out if a global optimization scheme is used for spontaneous sharding, i.e., $E[|P_i|] = o(N)$.
Here, $f$ will be $O(1)$ if the transaction rates between nodes follow Poisson distribution. In that case, for instance, if $p = O( \frac{\ln N}{N})$ which suggests $c = O(\ln N)$, each node will only need to acquire on average $O(\frac{\ln N \cdot \ln N}{\ln \ln N})$ chains.

Then, we give an example showing that even without global optimizations, a naive local optimization scheme could already result in the reduction of acquired chains in a small network. In order to show that the spontaneous sharding could be achieved in the worst case, we artificially construct an extreme network.


\begin{example}
We consider $N$ nodes $\{1,2,\ldots, N\}$ placed in a ring and each node transacts to the next $c$ nodes on its right and receives from the next $c$ nodes on its left. Each node is given an initial amount of value and will uniformly at randomly make transactions to the $c$ nodes. The frequency and the amount of the transaction follow Poisson and uniform distributions, respectively.
\end{example}

We run a simulation with our system for $N=10,15,20,25$ nodes with difference connectivity $c$. We simulate the communication between nodes with Netty\footnote{\url{http://netty.io/}} and the main chain with Tendermint \cite{tendermint}. We apply a naive smart transacting algorithm in which the sender simply checks his own transaction records for the information about the $P_i$ of his receivers and choose the sources accordingly. The implementation details and the source code of our system and the simulation can be found in \url{https://github.com/blockchain-lab/ScaleOutDistributedLedger}.
Fig.~\ref{fig:2} shows the scenarios of sharding in stable states. It can be observed that when $c$ is small, we have $g<N$, meaning spontaneous sharding is achieved. However, when $c$ is equal to or larger than $3,4,4,6$ for $10,15,20,25$ nodes, respectively, all nodes would acquire all chains. Then the throughput of our system is no better than other scalable blockchain systems. It is due to the naiveness of the algorithm that we use for local optimization, e.g., the current algorithm does not avoid combining sources with different proofs into a single transactions and does not distinguish between chains received only once and chains being continuously updated. With better optimization algorithms, we believe that smaller $g$ can be achieved for larger $c$.


\begin{figure}
  \centering
    \includegraphics[width=0.45\textwidth]{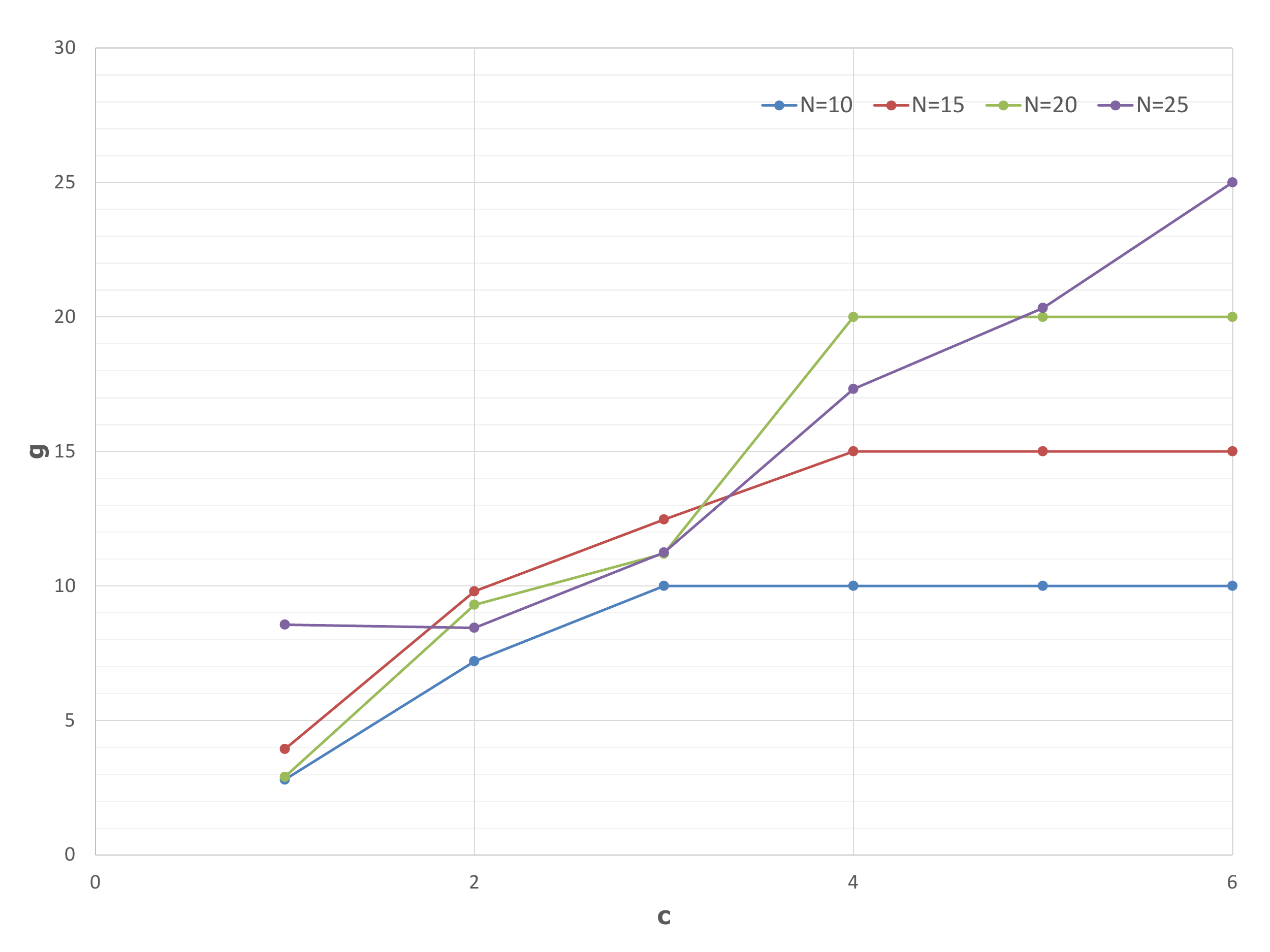}
    \caption{Average number of chains acquired by nodes for various connectivity in stable states.}\label{fig:2}
\end{figure}

\section{Conclusion and Future Work}\label{s:con}

In this paper, we proposed a novel blockchain system for the most considered type of distributed ledgers which we called VTL model. In VTL model, we assume that nodes are rational and will be motivated to prove their possessed value. Our system has a very simple and fully decentralized structure that does not introduce any node serving as ``validator''. Our system achieves uncompromised agreement and validity conditions and could scale out by spontaneous sharding without sacrificing security or decentralization. However, as the focus of this paper is put on the formal theoretical introduction of VTL and the off-chain and proof-based framework, many refinements in practical perspective are left for future research.

\begin{itemize}
\item {\bf Checkpoints to improve storage efficiency:} As sharding is a spontaneous and gradual process that might requires a initial phase, it might happened that nodes are required to record the whole transaction set until sharding starts. Then, the storage cost per transaction will not scale-out and cost new nodes quite heavily to join. This problem could be mitigated by introducing checkpoints in the main chain, which verifies the validity of certain values so that later on the proofs of these values do not have to date back to the genesis blocks. However, this do requires the newcomers to trust the old nodes who verified these values.
\item{\bf Private channels for low latency payments:} Private off-chain channels like \cite{lightning} is complicated in traditional blockchains since the notion of value is hinged to the on-chain ledger. However, as in our system the value is off-chain by nature, private channels for low latency micro-payments can be easily designed.
\item {\bf Supportive to conditional payments/smart contracts:} We conjecture that conditional payments and smart contracts can also be supported by this system with modified data structure and validation scheme as long as each transaction includes some value transferred to at least one of the receivers. Such system will simultaneously achieve sharding on both communication/storage resources and computation resources.
\item {\bf Real-world Implementation:} We conjecture that our system will also scale-out in more practical networks models \cite{banet,wsnet} or real-life transactions patterns \cite{wtw,ipf}. Moreover, we believe that for most of the cryptocurrencies nowadays, our system will be very beneficial in throughput since most of them are very ``trader-centric''. Then, for most users who only transact with traders, their transmission and storage cost can be significantly saved if the traders apply local optimization.
\item {\bf Discrimination and hidden forks:} As the proof size of transaction can be very different, it might cause issues of discrimination and hidden forks, e.g., values with huge proof sizes are refused by some nodes in the network, effectively causing a fork of the chain. This problem can be partially solved if checkpoints are used. However, we do not necessarily see this as a problem and argue that this is no more dangerous than forks in traditional blockchains.
\end{itemize}

\section*{Acknowledgment}

This work was supported by NWO Grant 439.16.614 Blockchain and Logistics Innovation.

The authors would also like to thank Chiel Bruin and Karol Jurasinski from Delft University of Technology for the help in implementing the scheme and the anonymous reviewers for the insightful comments.

\bibliography{Implicit_Consensus}
%
%






\fullv{\appendices

\section{Proof of Theorem~\ref{pro:ct}}\label{a:1}

\begin{proof}
\sloppy We proof this theorem by contradiction. If there exists a chain $\{B'_{u,1},B'_{u,2}, \ldots, B'_{u,k'}\}$ which consists of confirmed block and $t'_{u,k,l} \ne tx_i$. First of all, confirmed blocks suggest that all abstracts of $\{B'_{u,1},B'_{u,2}, \ldots, B'_{u,k'}\}$ is on the main chain. Then, since $t_{u,k,l}=tx_i$ is confirmed by an abstract $A_{u,k'}$, all abstract of the chain $\{B_{u,1},B_{u,2}, \ldots, B_{u,k'}\}$ are also on the main chain. Moreover, since the abstracts are signed and the digital signatures are assumed to be unbreakable, the chains $\{B'_{u,1},B'_{u,2}, \ldots, B'_{u,k'}\}$ and $\{B_{u,1},B_{u,2}, \ldots, B_{u,k'}\}$ will have the same set of abstracts of node $u$ on the main chain, which includes $A_{u,k'}$ as the abstract of both $B'_{u,k'}$ and $B_{u,k'}$. Since the hash function is assumed to be unbreakable, we have $B'_{u,k'}=B_{u,k'}$.
Then, if $t'_{u,k,l} \ne tx_i$, then we will have $B'_{u,k} \ne B_{u,k}$ and $B'_{u,k'}=B_{u,k'}, k' \geq k$, which contradicts our assumption of the unbreakable hash function. 
\end{proof}

\section{Proof Verification Algorithm}\label{a:2}

Here we give the Proof Verification Algorithm as Algorithm~\ref{alg:pv}.
\begin{algorithm}[h]
\caption{Proof Verification Algorithm $\mathsf{Ver}({\cal P}(tx_i)), tx_i \in B_{u,k}$}
\label{alg:pv}
\begin{algorithmic}
\State \#Verify the chain including this transaction
\State $\mathsf{count} \gets 0$
\State $\mathsf{absmark} \gets 0$
\For{$B_{s_i,m}, m=1:\max$} \Comment{Check the integrity of the chain}
	\If{$tx_i \in B_{s_i,m}$} $\mathsf{count} ++$ \EndIf
    \If{$m \ne 1$ {\bf and } first element in $B_{s_i,m} \neq H(B_{s_i,m-1})$} \Return $\mathsf{fail}$ \EndIf 
	\If{$A_{s_i,m}$ is included in the main chain}
    	\State $\mathsf{absmark} \gets m$
		\If{$H(B_{s_i,m}) \notin A_{s,m}$ {\bf or }$Sig_u(u ,k,H(B_{u,k}))$ is not correct}  \Return $\mathsf{fail}$ \EndIf
	\EndIf    
\EndFor
\If{$\mathsf{absmark} < k$} \Return $\mathsf{fail}$ \EndIf  \Comment{Check the confirmation}
\If{$\mathsf{count} \neq 1$} \Return $\mathsf{fail}$ \EndIf \Comment{Check the existence of the transaction}
\State \#Verify the chains of the sources
\For{all $tx_j \in {\rm Source}_i$} 
	\If{$\mathsf{Ver}(tx_j) \neq \mathsf{pass}$} \Return $\mathsf{fail}$ \EndIf
\EndFor
\Return $\mathsf{pass}$
\end{algorithmic}
\end{algorithm}

\section{Proof of Theorem~\ref{th:vs}}\label{a:4}
\begin{proof}
We first prove that if $\mathsf{V}(tx_i, {\cal P}(tx_i)) =\mathsf{valid}$ then $tx_i$ is valid. It directly follows from the four checks in Algorithm~\ref{alg:vf} since they are exactly the conditions in Definition~\ref{def:valtx}.


\sloppy We then show that if $\mathsf{V}(tx_i, {\cal P}(tx_i)) \ne \mathsf{valid}$ then $tx_i$ is not valid. To prove this, we first prove the statement ``if $\mathsf{V}(tx_i, {\cal P}(tx_i)) \ne \mathsf{valid}$ and $\forall tx_j \in {\rm Source}_i, \mathsf{V}(tx_j, {\cal P}(tx_j)) = \mathsf{valid}$, then $tx_i$ is not valid.''

\sloppy We prove this statement by contradiction. Assuming that there exists a transaction $tx_k$ such that $\mathsf{V}(tx_k, {\cal P}(tx_k)) \ne \mathsf{valid}$ but for all $tx_j \in {\rm Source}_k, \mathsf{V}(tx_j, {\cal P}(tx_j)) = \mathsf{valid}$, and $tx_k$ is valid.

By our algorithm, at least one of the four checks other than the ``{\bf Source Check}'' is failed. If the step ``{\bf Proof Check}'' fails, it suggests that a proof ${\cal P}(tx_i)$ does not exist, which contradicts the assumption that $tx_i$ is valid. If the step ``{\bf Equality Check}'' fails, it contradicts the {\bf Value equality} condition of valid transaction. If the step ``{\bf Double-Spending Check}'' fails, it contradicts the {\bf No double spending} condition of valid transaction.

We then prove that if $\mathsf{V}(tx_i, {\cal P}(tx_i)) \ne \mathsf{valid}$ then $tx_i$ is not valid by contradiction. If this does not hold, then there must exist a transaction that violates the statement proved above. This transaction might be $tx_i$, the source of $tx_i$, or recursively one in the sources of the sources.
\end{proof}

\section{Smart Transacting Algorithms}\label{a:3}
Here we give a naive smart transacting algorithms ${\rm Source}_i=\mathsf{ST}(d_i,a_i,{\cal C}_u)$ in Algorithm~\ref{alg:sta} for rational nodes, where node $u$ intends to send an amount of $a_i$ to node $d_i$ in transaction $tx_i$ and ${\cal C}_u$ is a collection of all transactions and proofs recorded in node $u$. Here, we assume that nodes have sufficiently large computation capability so that the computation cost is insignificant comparing to the communication cost.

\begin{algorithm}[h]
\caption{Non-interactive Smart Transacting Algorithm ${\rm Source}_i=\mathsf{ST}(d_i,a_i,{\cal C}_u)$}
\label{alg:sta}
\begin{algorithmic}
\State \#Step 1: Check for all unspent transactions
\State $\mathsf{UT} \gets$ all unspent $tx_i$ that are in ${\cal C}$.
\State \#Step 2: Determine the chains that $d$ already has according to ${\cal C}$
\State $\mathsf{Collected} \gets \emptyset$
\For{each $tx_i$ in ${\cal C}$ {\bf and} $d_i=d$}
	\State $\mathsf{chains}_i \gets \{v| {\cal B}_v \in {\cal P}(tx_i)\cap {\cal C}_u \}$ \Comment{All chains in the proof of $tx_i$ according to ${\cal C}_u$}
	\State $\mathsf{Collected} \gets \mathsf{Collected} \cup \mathsf{chains}(i)$
\EndFor

\State \#Step 3: Find the sources which has the least amount of chains to send
\For{all $\mathsf{Source}_n \subset \mathsf{UT}$ such that the sum amount no less than $a_i$}
	\State $\mathsf{Proof}_n \gets$ union of all ${\cal P}(tx_i), tx_i \in \mathsf{Source}_n$
    \State $\mathsf{NChains}_n \gets \{v| {\cal B}_v \in \mathsf{Proof}_n\}$
	\State $\mathsf{ToCollect}_n \gets \mathsf{NChains}_n / \mathsf{Collected}$
\EndFor
\Return $\mathsf{Source}_l$ where $\mathsf{ToCollect}_l= \min (|\mathsf{ToCollect}_n|)$
\end{algorithmic}
\end{algorithm}

Note that Algorithm~\ref{alg:sta} is a non-interactive algorithm. The choice of the sources is much easier in an interactive fashion, in which the receiver simply tells the sender the chains that he already has once per round. Then, the second step in Algorithm~\ref{alg:sta} can be omitted. The cost of this communication is $O(gc)$, where $g$ is average number of chains that a node acquires and $c$ is the average number of transacting targets of each node. This cost is no larger than $O(gN)$ and adds no more than $O(g)$ cost to the CCPT for the same reason that we have for global optimization in Subsection~\ref{ss:ss}. Both interactive and non-interactive algorithms will result in spontaneous sharding. For a stable network with sufficient transactions been made by each node, either interactive or non-interactive schemes will have similar performance since the transaction pattern is fixed and each node should already have enough prior knowledge for the chains that each receiver has. }

\end{document}